\documentclass{acm_proc_article-sp}

\begin{document}

\pagenumbering{arabic}

\title{On the Capacity of Rate-Adaptive Packetized Wireless Communication Links under Jamming \titlenote{Research partially supported by NSF Award CNS-0915985.}}
\subtitle{TO APPEAR IN Fifth ACM Conference on Security and Privacy in Wireless and Mobile Networks (WiSec'2012)}
\numberofauthors{3}
\author{
\alignauthor
Koorosh Firouzbakht\\
      \affaddr{Electrical and Computer}\\
      \affaddr{Engineering Department}\\
      \affaddr{Northeastern University}\\
      \affaddr{Boston, Massachusetts}\\
      \email{firouzbakht.k}\\
      \email{@husky.neu.edu}\\
\alignauthor
Guevara Noubir\\
      \affaddr{College of Computer and}\\
      \affaddr{Information Science}\\
      \affaddr{Northeastern University}\\
      \affaddr{Boston, Massachusetts}\\
      \email{noubir@ccs.neu.edu}\\
\alignauthor
Masoud Salehi\\
      \affaddr{Electrical and Computer}\\
      \affaddr{Engineering Department}\\
      \affaddr{Northeastern University}\\
      \affaddr{Boston, Massachusetts}\\
      \email{salehi@ece.neu.edu}\\
}
\maketitle

\markright{balalal}

\begin{abstract}
We formulate the interaction between the communicating nodes and an adversary within a
game-theoretic context. We show that earlier information-theoretic capacity results for a jammed
channel correspond to a pure Nash Equilibrium (NE). However, when both players are allowed to
randomize their actions (i.e., coding rate and jamming power) new mixed Nash equilibria appear with
surprising properties. We show the existence of a threshold ($J_{TH}$) such that if the jammer
average power exceeds $J_{TH}$, the channel capacity at the NE is the same as if the jammer was
using its maximum allowable power, $J_{Max}$, all the time. This indicates that randomization
significantly advantages powerful jammers. We also show how the NE strategies can be derived, and
we provide very simple (e.g., semi-uniform) approximations to the optimal communication and jamming
strategies. Such strategies are very simple to implement in current hardware and software.
\end{abstract}

%

\keywords{
Jamming, rate adaptation, capacity, game-theory.
}
\section{Introduction}	\label{Sec:Introduction}
Over the last decades, wireless communication proved to be an enabling technology to an
increasingly large number of applications. The convenience of wireless and its support of mobility
has revolutionized the way we access data, information services, and interact with the physical
world. Beyond enabling mobile devices to access information and data services ubiquitously,
wireless technology is  widely used in cyber-physical systems such as air-traffic control, power
plants synchronization, transportation systems, and human body implantable devices. This
pervasiveness elevated wireless communication systems to the level of critical infrastructure.
Radio-Frequency wireless communications occur over a broadcast medium, that is not only shared
between the communicating nodes but is also exposed to adversaries. Jamming is one of the most
prominent security threats as it not only can lead to denial of service attacks, but can also be
the prelude to spoofing attacks.

Anti - jamming has been an active area of research for decades. Various techniques for combating
jamming have been developed at the physical layer~\cite{SimonOSL01} which include directional
antennas, spread spectrum communication, power / modulation / coding control. At the time, most of the wireless communication were not packetized nor networked. Reliable communication in the presence of adversaries regained significant interest in the last few years, as new jamming attacks and the need for more complex applications and deployment environments have emerged. Several specifically crafted attacks and counter-attacks were proposed for packetized wireless data networks~\cite{negi+p:jam,LinN04,LiKP07,WilhelmMSL11}, 
multiple access resolution~\cite{BenderFHKL05,GilbertRC06,BayraktarogluKLNRT08, AwerbuchRS08}, multi-hop networks~\cite{XuMTZ06, TagueSNP08, LiKP07}, 
broadcast and control communication~\cite{KBKV06, ChiangH07, ChanLNT07, TagueLP07,LazosLK09,LiuPDL10,LiuLK11}, cross-layer resiliency~\cite{LinN05}, wireless sensor networks~\cite{XuKWY06,XuTZ07,XuTZ08}, spread-spectrum without shared secrets~\cite{StrasserCSM08, SlaterPRB09, StrasserCS09, JinNT09}, and navigation information broadcast systems~\cite{RasmussenCC07}. 

Nevertheless, very little work has been done on protecting rate adaptation algorithms against 
adversarial attacks. Rate adaptation plays an important role in widely used wireless 
communication systems such as IEEE802.11 standard as the link quality in a WLAN is often 
highly dynamic. In recent years, a number of algorithms for rate adaptation have been proposed in literature~\cite{HollandVB01,JuddXP08,VutukuruHK09,RahulFDC09,RamachandranKZG08,CampE08,KimSSD06,WongHSV06}, and some are widely deployed~\cite{Bicket05,LacageMT04}. 
Recently, rate adaptation for the widely used
IEEE 802.11 protocol was investigated in~\cite{PelechrinisBKG09, BroustisKDSL09, NoubirRST11}. Experimental and
theoretical analysis of optimal jamming strategies against currently deployed rate adaptation
algorithms indicate that IEEE 802.11 can be significantly degraded with very few interfering
pulses. The commoditization of software radios makes these attacks very practical and calls for
investigation of the capacity of packetized communication under adaptive jamming.

In this work, we focus on the problem of determining the optimal rate control and
adaptation mechanisms for a channel subject to a power constrained jammer. We consider a setup
where a pair of nodes (transmitter and receiver) communicate using data packets. An adversary
(jammer) can interfere with the communication but is constrained by an instantaneous maximum power
per packet ($J_{Max}$) and a long-run average power ($J_{Ave}$). Appropriately coded packets can
overcome interference and are lost otherwise. Over-coding (coding at low rates) reduces the
throughput, while under-coding (coding at high rates) increases the chances of loosing a packet. An
important question is to understand the interaction between the communicating nodes and the
adversary, determine the long-term achievable maximum throughput and the optimal strategy to
achieve it, as well as the optimal strategy for the adversary. While, the capacity of a channel
under a fixed-power jammer, and the optimal strategies for communication and jamming, derive from
fundamental information theoretic results (See Section~\ref{Sec:Game_Analysis}), these questions
are still open for a packetized communication system. 

Our contribution can be summarized as follows:

\begin{itemize}
\item We formulate the interaction between the communicating nodes and an adversary within a
game-theoretic context. We show the existence of the Nash Equilibrium for this non-typical game. 
We also show that the Nash Equilibrium strategies can be computed using Linear Programming.

\item We show that earlier information-theoretic capacity results for a jammed
channel correspond to a pure Nash Equilibrium (NE).

\item We further characterize the game by showing that, when both players are allowed to
randomize their actions (i.e., coding rate and jamming power) new mixed Nash equilibria appear with
surprising properties. We show the existence of a threshold ($J_{TH}$) such that if the jammer
average power exceeds $J_{TH}$, the channel capacity at the NE is the same as if the jammer was
using $J_{Max}$ all the time. 

\item We also show that the optimal NE strategies can be approximated by very simple (e.g.,
semi-uniform) distributions. Such strategies are very simple to implement in current hardware and software.
\end{itemize}

The rest of the paper is structured as follows. In Section \ref{Sec:System_Model}, we present our
model for the communication link, communicating nodes and the adversary. In Section
\ref{Sec:Game_Model}, we introduce the players, the \emph{transmitter} and the \emph{jammer}, and
their respective strategies and payoffs. We discuss how additional constraint on jammer's mixed
strategy space makes our game model different from a typical zero-sum game.
In Section~\ref{Sec:4}, we show that the Nash equilibrium indeed exists. We also prove the
existence of a threshold, $J_{TH}$, for the jammer and its effect on the game outcome. In
Section~\ref{Sec:Game_Analysis}, we study two particular cases. The case of a  \emph{powerful
jammer}, when jammer's average power is greater than the threshold, and the case of a \emph{weak
jammer}, when jammer's average power is less than the threshold. We will also provide transmitter's
optimal strategies in these two cases. In Section~\ref{Sec:Continuous}, we study the case where
players have infinite number of pure strategies (the continuous zero-sum game) and finally, we
conclude the paper in Section~\ref{Sec:Conclusion}.


\begin{table*}\label{Table:Parameters}
\centering \caption{Table of Notations and Parameters}
\renewcommand{\arraystretch}{1.5}
\begin{tabular}{|c|l|}
 \hline
  Parameter & Description \\
  \hline \hline
  $P_T$ & Transmitter's power\\
  $N$     & Noise power spectral density\\
  $J_{Max}$ & Jammer's maximum power per packet\\
  $J_{Ave}$ & Jammer's average power\\
  $J_{TH}$  & Jamming power threshold\\
  $J$       & Variable denoting jammer's power\\
  $J_T$     & Jamming power corresponding to the transmitter's rate \\ 
  \hline
$\boldsymbol{J}^T   = \begin{bmatrix} J_0 & \dots & J_j & \dots & J_{N_J} \end{bmatrix}_{1\times(N_J+1)}$ & Jamming power vector\\
  $J_j = \frac{j}{N_J}J_{Max}$ & \\
  \hline
  $\boldsymbol{R}^T = \begin{bmatrix} R_0 & \dots & R_i & \dots & R_{N_T} \end{bmatrix}_{1\times(N_T+1)}$ & Vector corresponding to transmitter's rates\\
  $R_i = \frac{1}{2}\log \left(1 + \frac{P_T}{N + \frac{i}{N_T}J_{Max}}\right)$ & \\
  \hline
  $ \boldsymbol{x}^T = \begin{bmatrix} x_0 & \dots & x_i & \dots & x_{N_T} \end{bmatrix}_{1\times(N_T+1)} \in \mathbb{X}$ & Transmitter's mixed-strategy vector\\
  $ \boldsymbol{y}^T = \begin{bmatrix} y_0 & \dots & y_j & \dots & y_{N_J} \end{bmatrix}_{1\times(N_J+1)} \in \mathbb{Y}$ & Jammer's mixed strategy vector\\
  $\mathbb{X,Y}$ & Mixed-strategy space, transmitter's and jammer's respectively\\
  $ C_{(N_T+1)\times(N_J+1)}\ \text{or}\ C(\boldsymbol{x},\boldsymbol{y})\ \text{or}\ C(J_{Ave})$ & Game matrix and expected game payoffs\\
  \hline
\end{tabular}
\end{table*}

\section{System Model}	\label{Sec:System_Model}

In this section we introduce and define our system model. The overall system model is shown in
Figure~\ref{Fig:Fig1}. The communication link between the transmitter and the receiver is an AWGN
channel with a fixed noise variance. Beside the channel noise, transmitted packets are being
disrupted by an additive jammer. Jammer's peak and average power are  assumed to be limited to
produce a more realistic model.

\subsection{Channel Model}\label{SubSec:Channel_Model}
The overall system model is shown in Figure~\ref{Fig:Fig1}. The communication link between the
transmitter and the receiver is assumed to be a single-hop, additive white Gaussian noise (AWGN)
channel with a fixed and known noise variance, $N$, referred to the receiver's front end.
Furthermore, the communication link is being disrupted by an additive adversary, the \emph{jammer}.
The jammer transmits radio signals to degrade the capacity between the transmitter and the
receiver. We assume transmissions are \emph{packet-based}, i.e., transmissions take place in
disjoint time intervals during which transmitter's and jammer's state (parameters) remain
unchanged. We assume packets are long enough that channel capacity theorem could be applied to each
packet being transmitted, this is justified by today's Internet protocols that use packet sizes of
up to $1,500$ bytes\footnote{IEEE 802.3 and IEEE 802.11x protocols allow MAC frame sizes of up to
1,642 and 2304 bytes respectively.}.

In section~\ref{Sec:Game_Model} we introduce and study a two-player zero-sum game in which
transmitter-receiver goal is to achieve highest possible rate while jammer tries to minimize the
achievable rate.

\begin{figure}
 \centering
 \includegraphics[width = 3.2in]{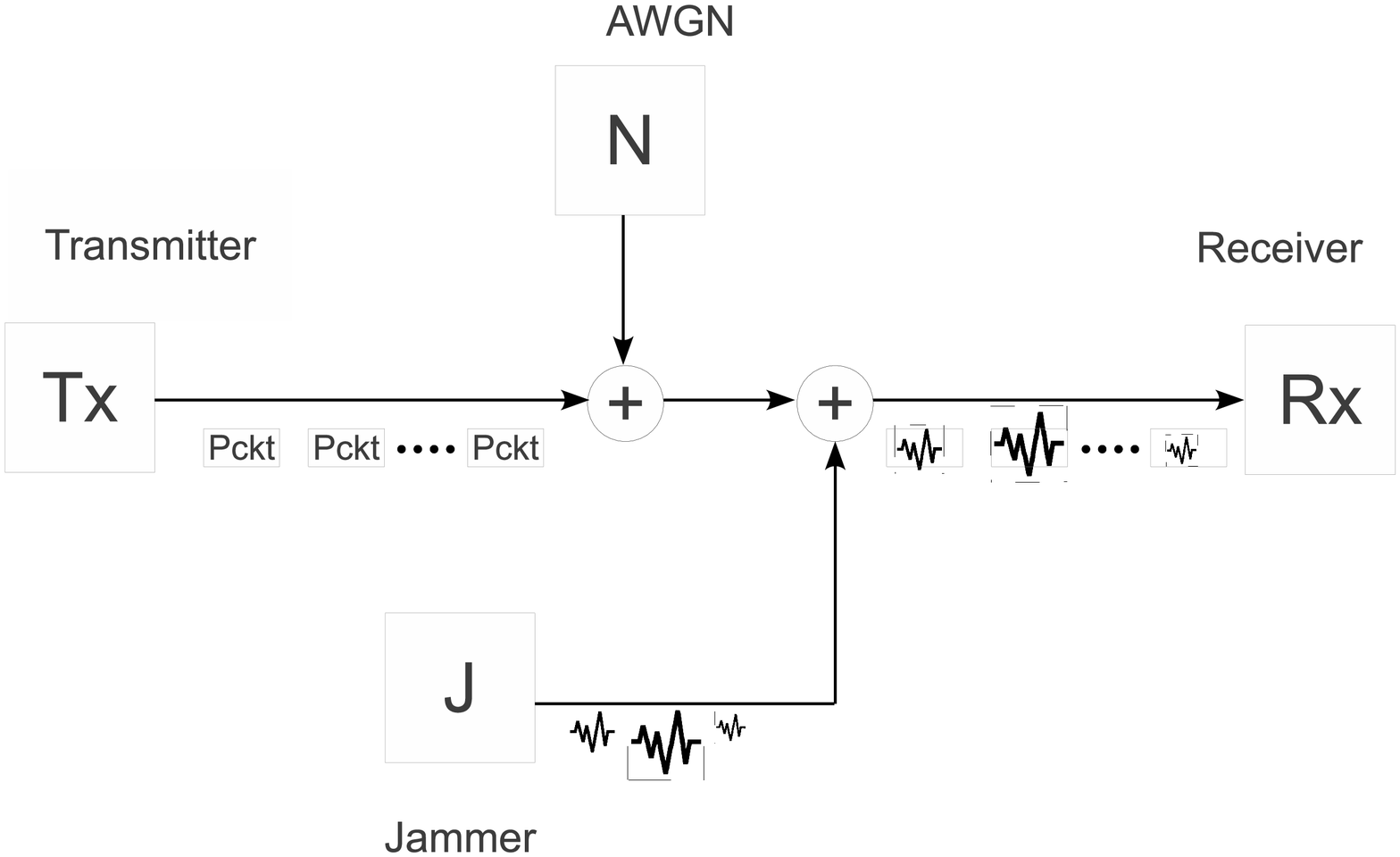}
 \caption{System model}
 \label{Fig:Fig1}
\end{figure}

\subsection{Jammer Model}\label{SubSec:Jammer_Model}
Radio jamming or simply \emph{jamming} is deliberate transmission of radio signals with the
intention of degrading a communication link. The effect of jammer on the communication link is
reduction of the effective signal to noise ratio (SNR) at the receiver and hence decreasing the
channel capacity. As long as reduction in effective signal to noise ratio is concerned, the jammer
can use arbitrary random signals for transmission but, it can be shown \cite{CoverJ06} that in the
AWGN channel with a fixed and known noise variance, a Gaussian jammer with a flat power spectral
density is the most effective in minimizing the the capacity between the transmitter and the
receiver. In other words, in the communication game described above, the optimal strategy for the
transmitter is to use a zero-mean white Gaussian input with variance equal to $P$, the transmitter
power, and the best strategy for the jammer is is to use a similar distribution with variance $J$,
the jammer power.

A fairly large number of jamming models have been proposed in the literature \cite{Peterson}. The
most benign jammer is the \emph{barrage noise jammer}.  The barrage noise jammer transmits
bandlimited white Gaussian noise with power spectral density (psd) of $J$. It is usually assumed
that the barrage noise jammer power spectrum covers exactly the same frequency range as the
communicating system. This kind of jammer simply increases the Gaussian noise level from $N$ to
$(N+J)$ at the receiver's front end. Another frequently used jamming model is the \emph{pulse-noise
jammer}. The pulse noise jammer transmits pulses of bandlimited white Gaussian noise having total
average power of $J_{Ave}$ referred to the receiver's front end. It is usually assumed that the
jammer chooses the center frequency and bandwidth of the noise to be the same as the transmitter's
center frequency and bandwidth. The jammer chooses its pulse duty factor to cause maximum
degradation to the communication link while maintaining the average jamming power $J_{Ave}$. For a
more realistic model, the pulse-noise jammer could be subject to a maximum peak power constraint.
Other jamming models, to name a few, are  the \emph{partial-band jammer} and
\emph{single/multiple-tune jammer}.

However, we study a more sophisticated jamming model. The jammer in study is a reactive and
additive jammer, i.e., he is only active when a packet is being transmitted and silent otherwise.
We assume that the jammer has a set of discrete jamming power levels uniformly distributed between
$J=0$ and $J=J_{Max}$. The jammer can choose any jamming power level given that he maintains an
overall average jamming power, $J_{Ave}$. The jammer uses  his available power levels according to
a distribution (his strategy), he chooses an optimal distribution to minimize the achievable
capacity of the communication link while maintaining his maximum and average power constraints,
i.e.,  $J_{Max}$ and $J_{Ave}$, respectively.

For reasons given in section \ref{SubSec:Transmitter_Model}, burst jamming (transmitting a burst of
 white noise to disrupt a few bits in a packet) is not an optimal jamming scheme. Hence, we assume
 the jammer remains active during the entire packet transmission, i.e., the jammer transmits a
 continuous Gaussian noise with a fixed variance $J \in \big[ 0, J_{Max} \big]$ for each transmitting
 packet.

\subsection{Transmitter Model}\label{SubSec:Transmitter_Model}

Transmitter has a rate adaptation block which enables him to transmit at different rates. Popular
techniques to increase or decrease the rate of a code are puncturing or extending. Puncturing and
extending increase the flexibility of the system without significantly increasing its complexity.
Considering jammer's activity, the transmitter changes his rate according to a distribution (his
strategy). Changing the rate can be accomplished using techniques like rate-compatible puncturing.
The transmitter chooses an optimal distribution to achieve the best possible average rate (payoff).
Same as before, we assume transmissions are \emph{packet-based}, i.e., transmissions are taken
place in disjoint time intervals during which, transmitter's rate remain unchanged. Transmitter's
model is shown in Figure~\ref{Fig:Transmitter}.

The interleaver block in transmitter's model is a countermeasure to burst errors and burst jamming.
Interleaving is frequently used in digital communications and storage devices to improve the burst
error correcting capabilities of a code. Burst errors are specially troublesome in short length
codes as they have very limited error correcting capabilities. In such codes, a few number of
errors could result in a decoding failure or an incorrect decoding. A few incorrectly decoded
codewords within a larger frame could make the entire frame corrupted.

Fortunately, combining effective interleaving schemes such as cryptographic interleaving and
capacity-achieving codes such as turbo codes and LDPC codes results in transmission schemes that
have good burst error correcting properties (see \cite{LinN04}) which make burst jamming ineffective.
Therefore, in our study we do not consider burst jamming and instead assume that the jammer remains
active during the entire packet transmission.

\begin{figure}
 \centering
 \includegraphics[width = 3.2in]{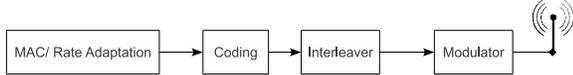}
 \caption{Transmitter Model}
 \label{Fig:Transmitter}
\end{figure}

\section{Game Model}	\label{Sec:Game_Model}

In this section we discuss the game setup in detail; we introduce and define the players, their
respective strategies and the constraints in the game. We present the game model and  define and
formulate the payoff function in a game theoretic frame work. As discussed in
section~\ref{SubSec:Jammer_Model}, in the AWGN channel, the additive white Gaussian jammer is the
optimal jammer, in the sense that the white Gaussian jammer minimizes the channel capacity.
Henceforward, we will only consider the additive Gaussian jammer.

We present the jammer's strategy set and introduce the jammer's average power constraint and its
impact on the mixed strategy space. The additional constraint makes our game model different from a
typical two-player zero-sum game. We also introduce transmitter's strategy set and define the game
utility function and the payoff matrix.

We begin by introducing a discrete version of the game to prove basic concepts and conclusions.
Generalization to the continuous case is given in Section~\ref{Sec:Continuous}

\subsection{\secit{The Jammer's} Strategy Set} \label{SubSec:Jammer}

The jammer has the option to select discrete values of jamming power, uniformly distributed over $
\big[ 0, J_{Max} \big]$. We assume there are $(N_J + 1)$  pure strategies available to the jammer.
Hence, the jammer's strategy set (set of jamming powers), $\mathcal{J}$, is given by
\begin{equation}\label{Eq:Jammer_Mixed_Set}
 \mathcal{J} = \Big\{J_j, 0 \leq j \leq N_J \Big\}
\end{equation}
where
\begin{equation}\label{Eq:Power_Levels}
 J_j = \frac{j}{N_J}J_{Max}
\end{equation}
We can write the possible jammer power levels in vector form, hence the jammer's pure strategies
vector, $\boldsymbol{J}$, is
\begin{equation}\label{Eq:Jammer_Vector}
  \boldsymbol{J}^T   = \begin{bmatrix} J_0 & \dots & J_j & \dots & J_{N_J} \end{bmatrix}_{1\times(N_J+1)}
\end{equation}
where $^T$ indicates transposition and $J_j$ is defined in \eqref{Eq:Power_Levels}. Unlike typical
zero-sum games in which there are no other constraints on the mixed-strategies, in our model, the
jammer's mixed-strategy must satisfy the additional average power constraint, $J_{Ave} \leq
J_{Max}$. Hence, in this model, not all mixed-strategies (and not even the pure strategies that are
greater than $J_{Ave}$) are feasible strategies \cite[Sec. III.7]{Owen}. If we let $\boldsymbol{y}$ be the jammer's
mixed-strategy vector and $\mathbb{Y}$ be the $(N_J + 1)$-simplex, we have the following relations:
\begin{align}
\boldsymbol{y}^T& = \begin{bmatrix} y_0 & \dots & y_j & \dots & y_{N_J} \end{bmatrix}_{1\times(N_J+1)} \in \mathbb{Y}\\
  \sum_{j=0}^{N_J} y_j &= 1; \quad y_j \geq 0, \quad 0\leq j \leq N_J \nonumber
\end{align}
By using the jammer's pure strategy vector we define the constrained mixed strategy space
$\mathbb{Y}_{\text{E}}$ as
\begin{equation}\label{Eq:Jammer_Mixed_Set_1}
 \mathbb{Y}_{\text{E}} = \{ \boldsymbol{y} \in \mathbb{Y} | \; \boldsymbol{y}^T \cdot
 \boldsymbol{J} = J_{Ave} \}
\end{equation}
which is a subset of the $(N_J+1)$-simplex that satisfies the average power constraint. By
substituting the equality constraint in \eqref{Eq:Jammer_Mixed_Set_1} with the less than or equal
sign, we define a new mixed strategy space which consists of all mixed strategies that result in an
average power less than or equal to $J_{Ave}$. The new mixed-strategy space,
$\mathbb{Y}_{\text{LE}}$, is
\begin{equation}\label{Eq:Jammer_Mixed_Set_2}
 \mathbb{Y}_{\text{LE}} = \{ \boldsymbol{y} \in \mathbb{Y} | \; \boldsymbol{y}^T \cdot
 \boldsymbol{J} \leq J_{Ave} \}
\end{equation}
It is obvious that
\begin{equation*}
  \mathbb{Y}_{\text{E}} \subset \mathbb{Y}_{\text{LE}} \subset \mathbb{Y}
\end{equation*}

\begin{figure}
 \centering
 \includegraphics[width = 3.2in]{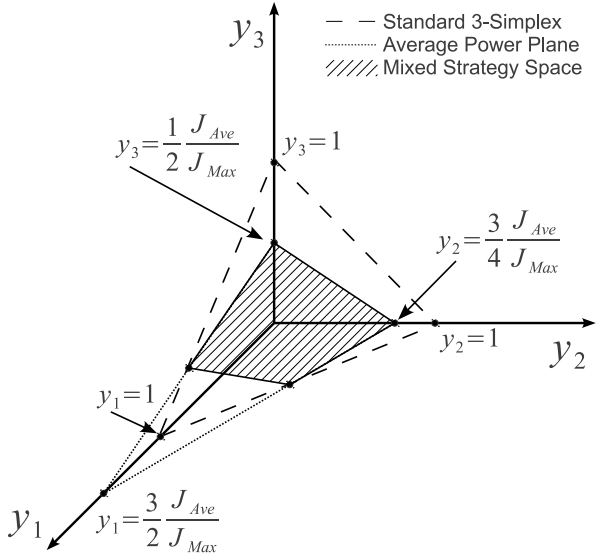}
 \caption{A typical mixed-strategy space for average power constrained jammer $(J_{Ave} < J_{Max})$.}
 \label{Fig:Jammer_Mixed_Space}
\end{figure}

A typical mixed strategy space with equality constraint, as defined in
\eqref{Eq:Jammer_Mixed_Set_1}, is shown in Figure~\ref{Fig:Jammer_Mixed_Space} where $N_J = N_T =
3$. In this case jammer's mixed and pure strategy vectors are $\begin{bmatrix} y_0 & y_1 & y_2 &
y_3 \end{bmatrix}_{1\times 4}$ and $\begin{bmatrix} 0 & \frac{1}{3}J_{Max} & \frac{2}{3}J_{Max} &
J_{Max} \end{bmatrix}_{1\times 4}$.

Since by introducing the new mixed strategy spaces of \eqref{Eq:Jammer_Mixed_Set_1} and
\eqref{Eq:Jammer_Mixed_Set_2} we are eliminating some mixed strategies that could have been
otherwise selected, the existence of the Nash equilibrium for this case must be first established.
This is unlike a typical zero-sum game with a finite number of pure strategies in which the
existence of the Nash Equilibrium is assured. In section~\ref{SubSec:Nash_Existence}, we provide an
outline of the proof of the existence of the Nash Equilibrium in our game where the jammer's mixed
strategy space is limited to $\mathbb{Y}_{\text{E}}$ or $\mathbb{Y}_{\text{LE}}$.

\subsection{\secit{The Transmitter's} Strategy Set}\label{SubSec:Transmitter}

The transmitter strategy set is a set of discrete transmission rates corresponding to different
assumed jamming power levels, i.e, the transmitter chooses his rate, $R$, from the set
\begin{equation}
\mathcal{R} = \Big\{ R_i, 0 \leq i \leq N_T  \Big\}
\end{equation}
where
\begin{equation}\label{Eq:Rate}
  R_i = \frac{1}{2}\log \left( 1 + \frac{P_T}{N+\frac{i}{N_T}J_{Max}} \right)
\end{equation}
and $\frac{i}{N_T}J_{Max}$ denotes the jammer's power level assumed by the transmitter. If the
actual jammer's power level is less than or equal to the assumed value of $\frac{i}{N_T}J_{Max}$,
then transmission at rate $R_i$ is possible, otherwise reliable transmission is not possible, the
packet is lost, and the actual transmission rate  drops to zero. Same as the case with the jammer,
we define the vector of mixed-strategies for the transmitter, $\boldsymbol{x}$, as
\begin{equation}
 \boldsymbol{x}^T = \begin{bmatrix} x_0 & \dots & x_i & \dots & x_{N_T}
 \end{bmatrix}_{1\times(N_T+1)} \in \mathbb{X}
\end{equation}
where $\mathbb{X}$ is the $(N_T+1)$-simplex with no additional constraints.

\subsection{The Payoff Function}\label{SubSec:Payoff}
The \emph{payoff} to the transmitter is defined assuming transmissions at the channel capacity.
Defining the payoff based on channel capacity (or other variations of channel capacity) is a common
practice in the games involving a transmitter-receiver pair and an adversary \cite{Giannakis08,
Koorosh11,Altman07}.

Because transmissions occur in the presence of an adversary, recovery of the transmitted
information at the receiver is not always guaranteed. The information can only be recovered when
the actual jamming power, $J$, is less than or equal to the jamming power level assumed by the
transmitter, $J_T$, i.e., if and only if $J \leq J_T$. If $J_T < J$, the corresponding transmission
rate would exceed the channel capacity and the information would be lost. Therefore, the
transmitter's payoff function is given by
\begin{equation}\label{Eq:Payoff}
 C(J_T,J) =
    \begin{cases}
      R(J_T) = \frac{1}{2}\log \left( 1 + \frac{P_T}{N+J_T} \right) & J_T \geq J\\
      0 & J_T<J\\
    \end{cases}
\end{equation}
Since the game in study is a zero-sum game, the payoff to the jammer is the negative of the
transmitter's payoff. We can formulate the payoffs in a payoff matrix where the transmitter and the
jammer would be the row and column players respectively. The resulting payoff matrix, $C$, is
\begin{equation}\label{Eq:Game_Matrix}
  C =
  \begin{bmatrix}
    R_0     & 0       & 0     & \dots   & 0      \\
    \vdots  & \ddots  & 0     & 0       & \vdots \\
    R_i     & \dots   & R_i   & 0       & \vdots \\
    \vdots  &         &       & \ddots  & \vdots \\
    R_{N_T} & R_{N_T} & \dots & R_{N_T} & R_{N_T}\\
  \end{bmatrix}_{(N_T+1)\times(N_T+1)}
\end{equation}
where $R_i$ is defined in \eqref{Eq:Rate}. The expected payoff (or the game value) of the game is
\begin{equation}\label{Eq:Expected_Payoff}
  C(\boldsymbol{x},\boldsymbol{y}) = \boldsymbol{x}^T \cdot C \cdot \boldsymbol{y}, \qquad
  \boldsymbol{y} \in \mathbb{Y}_{\text{E}}\ \text{or}\ \mathbb{Y}_{\text{LE}}
\end{equation}

\newtheorem{lem}{Lemma}

In defining \eqref{Eq:Game_Matrix} we have assumed $N_J=N_T$. As discussed below, without loss of
generality, we can always assume that $N_T = N_J$.

\begin{figure}
 \centering
 \includegraphics[width = 3.2in]{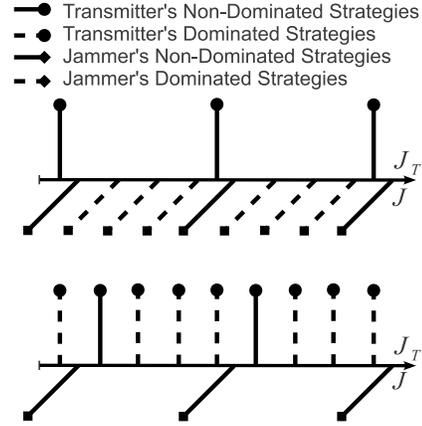}
 \caption{Uniformly distributed pure strategies; $N_J > N_T$ (top) and $N_J < N_T$ (bottom)}
 \label{Fig:Dominate_Strategies}
\end{figure}

\begin{lem}\label{Lem:N_T}
  Let $C$ be the payoff matrix in the two-player zero-sum game defined by the utility function
  \eqref{Eq:Payoff}. The payoff matrix resulted by removing the dominated strategies is a square lower
  triangular matrix with  size less than or equal to $\min \big[ N_T, N_J \big]$.
 Furthermore, if the power levels were uniformly distributed over $\big[ 0, J_{Max} \big]$,
  the size of the non-dominated payoff matrix would be the minimum of $N_T$ and $N_J$.
\end{lem}

\begin{proof}
Assume the jammer's power levels are arbitrary distributed over some range, $\big[ 0, J_{Max}
\big]$, and $N_T < N_J$.
A typical case where
$N_T < N_J$ is depicted in Figure \ref{Fig:Dominate_Strategies} (top). In Figure
\ref{Fig:Dominate_Strategies}, the transmitter's pure strategies are mapped to the jammer's power
levels for better visualization. Between some of the transmitter's pure strategies there might be a
pure strategy of the jammer but since $N_T < N_J$, according to the \emph{Pigeonhole} principle,
between at least two of the transmitter's pure strategies (not necessarily any two pure strategy as
sketched) there must be  more than one jamming power level (shown as dashed or solid lines ending
in squares). Any of these jamming power levels (or pure strategies) could be used to terminate the
information transmitted by the rate corresponding to the power level immediately to the left of
them (shown as solid line ending in circles). From these pure strategies, a rational jammer would
choose the one with the lowest power level (the solid line) and hence, it would dominate the rest
(dashed lines). therefore, the number of non-dominated pure strategies for the jammer is at most
equal to the the number of the transmitter's pure strategies (first part of the lemma).

If the pure strategies were uniformly distributed over $\big[ 0,$ $J_{Max} \big]$, as
sketched, for every transmitter's pure strategy there would be exactly one non-dominated strategy
for the jammer and hence, there would be no intention for the jammer to use more pure strategies
than the transmitter. The same discussion can be given for the number of pure strategies a rational
transmitter should use for the case $N_T > N_J$ (see Figure \ref{Fig:Dominate_Strategies}(
bottom)). Henceforward, without loss of generality, we assume $N_T = N_J$.
\end{proof}

As a consequence of  Lemma \ref{Lem:N_T}, in our study, we need to consider only square matrices
which simplifies further studies and assumptions. In the section that follows, we will study the
outcome of the game when jammer's average power assumes different values.


\section{Game Characterization}	\label{Sec:4}

In this section, we study the basic properties of the game. We will show that although we have put
an additional constraint on the jammer's mixed strategy space, the existence of the Nash
equilibrium is still guaranteed.

Furthermore, we will show that by randomizing his strategy, the jammer can force the transmitter to
operate at his lowest rate, given that he uses an average jamming power, $J_{Ave}$, that is more
than a certain threshold, $J_{TH} < J_{Max}$. We also provide an upper bound for $J_{TH}$ in this
section.

\subsection{Existence of the Nash Equilibrium}\label{SubSec:Nash_Existence}
We begin this section by the following lemma that shows existence of the Nash equilibrium under the
additional average power constraint is guaranteed.
  \begin{lem}
   For the two-player zero-sum game defined by the utility function $C \left( J_T,J \right)$,
   given in \eqref{Eq:Payoff} and the payoff matrix $C$, given by \eqref{Eq:Game_Matrix} and
   the transmitter's mixed strategy, $ \boldsymbol{x} \in \mathbb{X} $, and the jammer's mixed
   strategy, $\boldsymbol{y} \in \mathbb{Y}_{\text{E}}\ \text{or}\ \mathbb{Y}_{\text{LE}}$ (defined
   in \eqref{Eq:Jammer_Mixed_Set_1} and \eqref{Eq:Jammer_Mixed_Set_2}, respectively), at least one Nash
   equilibrium exists.
  \end{lem}
Nash in his 1951 seminal paper, ``Non Cooperative Games'' \cite{Nash51}, proved that for any game
with finite set of pure strategies, there exists at least one  (pure or mixed) equilibrium such
that no player can do better by unilaterally deviating from his strategy. In the proof of the
existence of the Nash equilibrium, no additional constraints were assumed on the mixed strategy
spaces. But, in our game model, we are assuming an additional constraint on the jammer's mixed
strategy space; the jammer must maintain a fixed or maximum average jamming power (corresponding to
\eqref{Eq:Jammer_Mixed_Set_1} and \eqref{Eq:Jammer_Mixed_Set_2}, respectively). These additional
assumptions change the jammer's mixed strategy space from the n-simplex to a subset of it.
Therefore, the Nash equilibrium theorem cannot be applied to our model directly and the existence
of the Nash equilibrium must be established.

\begin{proof}[(Outline)]
The proof of the existence of Nash equilibrium hinges on the \emph{Sperner's lemma} and 
\emph{Brouwer's fixed point theorem} and a corollary of this theorem on simplotopes \footnote{
There are alternative proofs for the existence of the Nash equilibrium, i.e., using \emph{Kakutani
fixed point theorem} \cite{Owen}. }. Sperner's lemma applies to simplicially subdivided n-simplexes. It can
easily be shown that by using a \emph{radial projection}, the mixed strategy space in our model,
which is a result of additional constraint of maintaining an average jamming power (or maintaining
a maximum average power), can be projected to an appropriate lower dimension  m-simplex where $m <
n$. A similar argument can be used to generalize the Brouwer's fixed point theorem to any arbitrary
convex and compact set. Since the additional average power constraint does not effect the convexity
or compactness of the mixed strategy space, we can conclude that all the conditions and
requirements assumed by the Sperner's lemma and the Brouwer's fixed point theorem are satisfied \cite{Basar99} and
the existence of the Nash equilibrium for our problem is guaranteed.
\end{proof}

\subsection{Existence of Jamming Power Threshold}\label{SubSec:Power_Threshold}

\begin{figure}
 \centering
 \includegraphics[width = 3.2in]{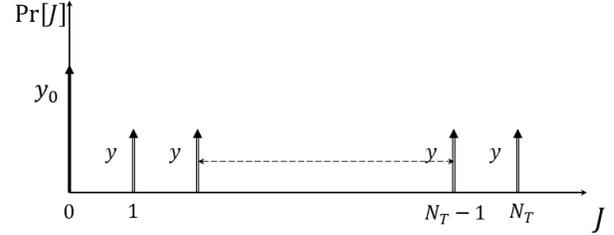}
 \caption{Semi-Uniform Distribution}
 \label{Fig:Uniform}
\end{figure}

\begin{figure*}
 \centering
 \includegraphics[width = 6in]{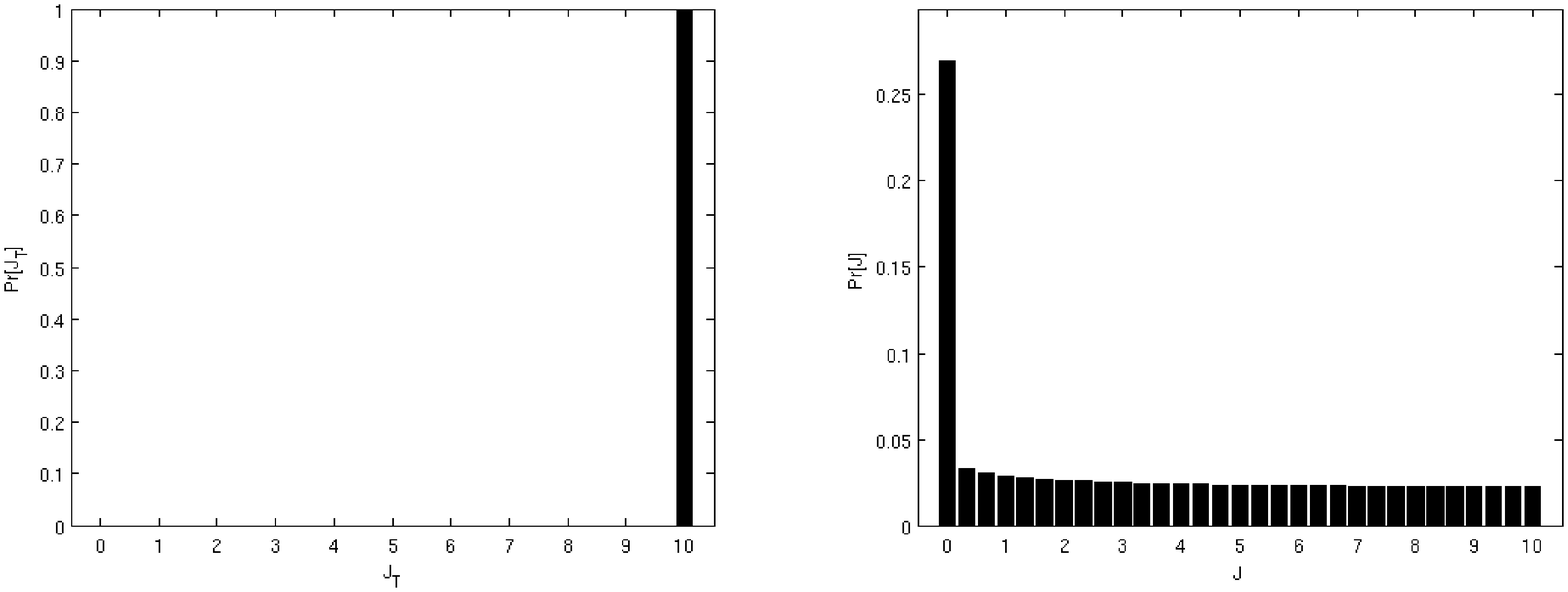}
 \caption{Typical optimal mixed-strategies for the transmitter (left) and the jammer (right) for
 $J_{Ave} \geq J_{TH} $} \label{Fig:Powerful_Jammer}
\end{figure*}

The following theorem proves the existence of a threshold jammer power that plays an important role
in our further development.
\newtheorem{theorem}{Theorem}
\begin{theorem}\label{The:Threshold}
  For the two-player zero-sum game defined with the utility function $C \left( J_T,J \right)$,
  given in \eqref{Eq:Payoff}, and the payoff matrix $C$, given in \eqref{Eq:Game_Matrix}, and the
  transmitter's mixed strategy, $ \boldsymbol{x} \in \mathbb{X} $, and the jammer's mixed strategy
  $\boldsymbol{y} \in \mathbb{Y}_{\text{LE}}$, given in \eqref{Eq:Jammer_Mixed_Set_2} and for all
  $P_T,N,J_{Max}~>~0 $
\begin{equation*}
 	 \exists J_{TH};\quad 0 < J_{TH} < J_{Max}
 \end{equation*}
such that, if $J_{Ave} \geq J_{TH}$ then, $\exists \boldsymbol{y}^* \in \mathbb{Y}_{LE}$ for which
we have
\begin{equation}
  \begin{aligned}
    & \boldsymbol{x}^{*^T} = \begin{bmatrix} \boldsymbol{0}_{1\times N_T} &  1 \\ \end{bmatrix}_{1\times(N_T+1)}\\
    & C(\boldsymbol{x}^*,\boldsymbol{y}^*) = R_{N_T}
  \end{aligned}
\end{equation}
where $\boldsymbol{x}^*,\boldsymbol{y}^*$ are transmitter's and jammer's optimal mixed-strategies, 
respectively and $ C(\boldsymbol{x}^*,\boldsymbol{y}^*) $ represents the value of the game.
\end{theorem}

Theorem~\ref{The:Threshold} states that there exists a jamming threshold ($J_{TH}$) such that if
the jammer's average  power exceeds $J_{TH}$ then the transmitter's optimal mixed-strategy is to
use the lowest rate.

\begin{proof}
Assume the jammer is using a mixed strategy with the pmf given in Figure \ref{Fig:Uniform}
(\emph{semi-uniform})\footnote{ We will refer to this class of pmf/pdf as the \emph{semi-uniform}}
which is not necessarily an optimal mixed strategy. The parameters of this pmf are
\begin{equation}
 \begin{aligned}
    y_0 & = 1 - \frac{2N_T}{N_T+1}\cdot \frac{J_{Ave}}{J_{Max}}\\
    y   & = \frac{2}{N_T+1}\cdot \frac{J_{Ave}}{J_{Max}}\\
 \end{aligned}
\end{equation}
It can be easily verified that the semi-uniform pmf satisfies the average power constraint
\begin{equation*}
  \begin{split}
    & \sum^{N_T}_{j=0} J\cdot \text{Pr} \big[ J \big]\\
    &  =  \sum^{N_T}_{j=0} \left( \frac{j}{N_T} J_{Max} \right) \cdot \text{Pr} \bigg[J = \left(
    \frac{j}{N_T} J_{Max} \right) \bigg]\\
    & = J_{Ave}
  \end{split}
\end{equation*}
We assume the transmitter is using an arbitrary mixed strategy in which rates $R_{N_T}$ (the lowest
rate corresponding to $J_T = J_{Max}$) and $R_i$ (an arbitrary rate corresponding to $J_T = $
$\frac{i}{N_T} J_{Max},$ $ 0\leq i < N_T$) have probabilities $x_{N_T}$ and $x_i$ respectively.
Define $C$ to be the expected payoff for the jammer's semi-uniform mixed strategy against the
transmitter's arbitrary mixed-strategy:
\begin{equation}\label{Eq:C_Payoff}
  \begin{split}
    C   & = C_{-i,N_T} + R_{N_T} x_{N_T} \times 1 + R_i x_i \times \text{Pr} \Big[ J \leq J_T = J_i \Big]\\
	& = C_{-i,N_T} R_{N_T}x_{N_T} + R_i x_i \left( y_0 + iy \right)\\
  \end{split}
\end{equation}
where $C_{-i,N_T}$ is the partial expected payoff resulting from all pure strategies except for the
$i$'th and $N_T$'th strategies. In order to improve his payoff, the transmitter, deviates from his
current strategy to $x'_{N_T} = x_{N_T} + \delta$ and $x'_j = x_j - \delta$ where $\delta >0$.
Defining $C'$ to be the expected payoff for the new strategy, we have
\begin{equation}\label{Eq:C_Prime_Payoff}
  \begin{split}
      C' &=   C_{-i,N_T} + R_{N_T} \left( x_{N_T} + \delta \right)\\
	   & \phantom{=} + R_i \left( x_i - \delta \right) \times \text{Pr} \Big[ J \leq J_T = J_i \Big]\\
	&= C + \delta \Big[ R_{N_T} - R_i \left( y_0 + iy \right) \Big]\\
  \end{split}
\end{equation}
Let $\Delta C $ be the difference in the expected payoff caused by deviating to the new strategy
\begin{equation}\label{Eq:Delta_C_InEq}
  \begin{split}
      \Delta C  & = C' - C\\
	        & = \delta \left( R_{N_T} - 2 R_i \frac{N_T - i}{N_T + 1} \cdot \frac{J_{Ave}}{J_{Max}} \right)\\
  \end{split}
\end{equation}
where $\delta>0$ and $0\leq i < N_T$. We show that there exists a jammer power threshold, denoted
by $J_{TH}$, such that if $J_{Ave} \geq J_{TH}$, then for all $\delta > 0$  and for all $i \in [0,
J_{Max})$, we have
\begin{equation}\label{Eq:Delta_C_Pos}
 \Delta C >0
\end{equation}
Assuming (for now) that $\Delta C > 0$ we can rewrite \eqref{Eq:Delta_C_InEq} as
\begin{equation}\label{Eq:Z}
  \begin{split}
    J_{Ave} &\geq \frac{1}{2}J_{Max} \frac{N_T + 1}{N_T - i} \cdot \left( 1 - \frac{R_{N_T}}{R_i} \right)\\
	               & = Z_i, \qquad 0\leq i < N_T\\
  \end{split}
\end{equation}
where $Z_i$'s, for $ i = 0, \dots ,N_T - 1$, are a set of $N_T$ finite values. Let us define
$J_{TH} = \max Z_i$, then for
\begin{equation}\label{Eq:J_TH_Upper_Bound}
 J_{Ave} \geq J_{TH}
\end{equation}
and for all $\delta > 0$ and $i\in [0, N_T)$ the inequalities in \eqref{Eq:Z} and
\eqref{Eq:Delta_C_Pos} are satisfied.

We showed that for $J_{Ave} \geq J_{TH}$, the transmitter can improve his expected payoff by
dropping probability from any arbitrary rate (except for the lowest rate) and adding this
probability to the lowest rate. We can continue this process until all other probabilities are
added to the lowest rate probability and no further improvement to the expected payoff is possible.
This shows that the low rate is indeed an optimal strategy for the transmitter against the jammer's
semi uniform mixed strategy.

By using the semi-uniform pmf and $J_{Ave} \geq J_{TH}$, the jammer can force the transmitter to
operate at the lowest rate and given that the expected payoff is bounded between the transmitter's
lowest and highest rates, we can conclude that the semi-uniform distribution is indeed an optimal
mixed strategy for the jammer when \eqref{Eq:Jammer_Mixed_Set_2} is the mixed strategy
space\footnote{The $J_{TH}$ given by \eqref{Eq:J_TH_Upper_Bound} is not necessarily the lowest
possible threshold since we have limited jammer's strategies to semi-uniform distributions.
However, it is an upper bound for the lowest $J_{TH}$.}.\end{proof}

It is interesting to note that the packetized transmission model employed here and the
transmitter's lack of knowledge of the actual jammer power level benefits the jammer. In fact, the
jammer uses a power level less than $J_{Max}$ but forces the transmitter to transmit at a rate
corresponding to $J_{Max}$. This is similar to the situation in fading channels where although the
ergodic capacity can be large, the outage capacity is considerably lower.

It is shown in Appendix \ref{App:I} that $Z_i$ in \eqref{Eq:Z} is maximized for $i=0$
. Therefore an upper bound for
$J_{TH}$ is
\begin{equation}\label{Eq:J_TH_Upper_Bound2}
 J_{TH,U} = \frac{1}{2} \frac{ N_T + 1 }{N_T} \left( 1 - \frac{R_{N_T}}{R_0} \right) J_{Max}
\end{equation}

In section \ref{SubSec:Powerful_Jammer} we show that by using an optimal mixed strategy, the jammer
can achieve a lower threshold than \eqref{Eq:J_TH_Upper_Bound2}.

\section{Game Analysis}	\label{Sec:Game_Analysis}

%
%
%

In this section we study the optimal mixed strategies for the jammer and the transmitter. We
provide analytic and computer simulated results and a comparison between power thresholds resulted
from computer simulation and the upper bound derived in section \ref{Sec:4}.

Based on relative values of  $J_{Ave}$ and $J_{TH}$, we study two cases, the \emph{powerful jammer}
where $J_{Ave} \geq J_{TH}$ and the \emph{weak Jammer} where$J_{Ave} < J_{TH}$.

\subsection{Powerful Jammer}\label{SubSec:Powerful_Jammer}

As a result of the Theorem \ref{The:Threshold}, there exists a jamming threshold ($J_{TH}$), such
that if the jammer's average power exceeds $J_{TH}$, then the transmitter's optimal mixed strategy
(or more accurately, the optimal pure strategy in this case) is to use the lowest rate. We
formulate this fact in the following theorem.

\begin{theorem}\label{The:Powerful_Jammer}
There exists a threshold $J_{TH}$ such that if $J_{Ave} \geq J_{TH}$, the expected payoff of the
game is
\begin{equation*}
  \begin{split}
  C\Big( J_{Ave} \Big) & =  R_{N_T}=\frac{1}{2} \log \left( 1 + \frac{P_T}{N+J_{Max}} \right)
  \end{split}
\end{equation*}
The value of $J_{TH}$ is given by
\begin{equation}\label{Eq:Power_Threshold}
 J_{TH} =  \left( 1 - \frac{1}{N_T}\alpha^{-1} R_{N_T} \right) J_{Max}
\end{equation}
where $R_i$ is defined in \eqref{Eq:Rate} and
\begin{equation}
  \alpha^{-1} = \sum^{N_T - 1}_{i = 0} \left( R_i \right) ^{-1}
\end{equation}
\end{theorem}

In other words, if the average jamming power exceeds $J_{TH}$ given in \eqref{Eq:Power_Threshold},
by randomizing his strategy, the jammer forces the transmitter to operate at his lowest rate as if
the jammer was using $J_{Max}$ all the time (Barrage noise jammer). If we define the effective
jamming power, $J_{\text{Eff}}$, to be the jamming power a Barrage noise jammer needs to force the
transmitter to operate at the same rate ($R_{N_T}$ in this case) then, for the powerful jammer the
effective jamming power becomes
\begin{equation}
 J_{\text{Eff}} = J_{Max}
\end{equation}

Typical optimal mixed strategies for the transmitter and the jammer in a powerful jammer case are
given in Figure~\ref{Fig:Powerful_Jammer}. Proof of Theorem \ref{The:Powerful_Jammer} is similar to
the proof of Theorem \ref{The:Threshold}. Details of deriving relation \eqref{Eq:Power_Threshold}
are given in Section \ref{SubSec:Weak_Jammer}.

Unfortunately, jammer's optimal mixed strategy cannot be formulated in a closed form relation and
the optimal distribution has to be calculated numerically. As we showed in section
\ref{SubSec:Power_Threshold}, the simple semi-uniform pmf, shown in Figure~\ref{Fig:Uniform}, could
be used to derive an upper bound for the jamming threshold and as an approximation to the jammers
optimal mixed strategy (see Figure \ref{Fig:Powerful_Jammer} (right)). The price paid by deviating
from the optimal mixed strategy to the simple semi-uniform distribution is that the jammer has to
use more average power to force the transmitter to operate at the lowest rate. A comparison between
the jammer's average power threshold given in \eqref{Eq:Power_Threshold} and the upper derived in
\eqref{Eq:J_TH_Upper_Bound2} is given in Figure~\ref{Fig:Comparison}.

\begin{figure}
 \centering
 \includegraphics[width = 3in]{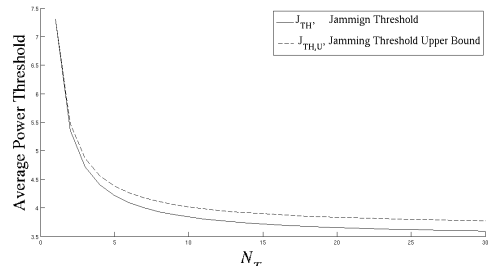}
 \caption{Comparison between the average power threshold and its upper bound }
 \label{Fig:Comparison}
\end{figure}

\subsection{Weak Jammer}\label{SubSec:Weak_Jammer}

\begin{figure*}
 \centering
 \includegraphics[width = 6in]{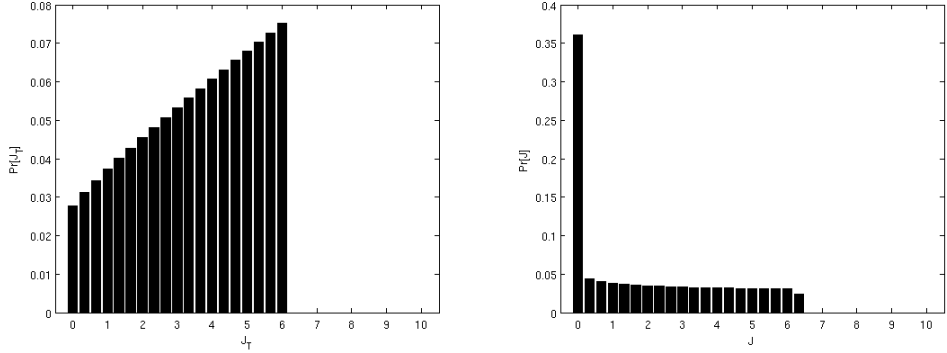}
 \caption{Typical optimal mixed-strategies for the transmitter (left) and the jammer (right) for
 $J_{Ave} < J_{TH} $} \label{Fig:Weak_Jammer}
\end{figure*}

A weak jammer has an average jamming power less than the threshold, $J_{Ave} < J_{TH}$. Typical
optimal mixed strategies for the weak jammer case are given in Figure \ref{Fig:Weak_Jammer}.

In this case the expected payoff, $C\Big( J_{Ave} \Big) \in \big( R_{N_T}, R_0 \ \big]$. Although a
useful closed form relation between the expected payoff and the jammer's average power where
$J_{Ave} \in \big[ 0,$ $ J_{TH} \big)$ cannot be derived, for specific values of the average
jamming power the relation reduces to a simple form. For these specific values, the expected payoff
of the game, $C\Big( J_{Ave} \Big)$, corresponds to one of the transmitter's rates $R_i$, $i =
0,\dots,N_T-1$. We present this fact in the following theorem without providing the full proof. 

\newtheorem{Conjecture}{Conjecture}
\begin{theorem}\label{The:Weak_Jammer}
 Assuming $ J_{Ave} < J_{TH}$
\begin{enumerate}
 \item
  The expected payoff of the game is
  \begin{equation}
    \begin{split}
      C\Big( J_{Ave} \Big) = & R_{m+1}\\
			  =& \frac{1}{2} \log \left ( 1 + \frac{P_T}{N + \frac{m + 1}{N_T}J_{Max}} \right)\\
    \end{split}
  \end{equation}
  where $m$ is the solution of
  \begin{equation}\label{Eq:J_AVE_M}
    J_{Ave} = \left( m + 1 - \alpha^{-1}R_{m+1} \right) \frac{J_{Max}}{N_T}
  \end{equation}

\item
  The transmitter's optimal mixed strategy is
  \begin{equation*}
    \boldsymbol{x}^{*^T}_m = \begin{bmatrix} x_0 & x_1 & \dots & x_m & 0 & \dots & 0
    \end{bmatrix}_{1\times(N_T+1)}
  \end{equation*}
where
\begin{equation}\label{Eq:Trans_Optimal_Weak_Jammer}
  x_i =\text{Pr} \left[ J_T = \left( \frac{i}{N_T} \right) J_{Max} \right]= \alpha_m R_i^{-1},\quad 0\le i\leq m
  \end{equation}
 and
  \begin{equation}
    \alpha^{-1}_m = \sum^{m}_{i = 0} \left( R_i \right) ^{-1}
  \end{equation}
\end{enumerate}
\end{theorem}

The optimal mixed strategies for a typical zero-sum two-player game could be calculated by
\emph{linear programming}. Our game model differs from a typical zero-sum game however, linear
programming could still be used to calculate the optimal mixed strategies by making the proper
modifications \cite{Owen} and even though we do not provide the full proof for
the transmitter's optimal mixed-strategy, the consistency of \eqref{Eq:Trans_Optimal_Weak_Jammer} can
be verified by computer simulation. Numerical calculations verify that results achieved by using
\eqref{Eq:Trans_Optimal_Weak_Jammer} as the transmitter's optimal mixed strategies are accurate to
the order of $10^{-15}$.

In order to prove \eqref{Eq:J_AVE_M}, we first introduce the following lemma without a proof.

\begin{lem}\label{Lem:Semi_Uniform_Approx}
 The semi-uniform distribution and the jammer's optimal mixed strategy (see
 Figure \ref{Fig:Weak_Jammer} (left)) result in the same expected payoff against the transmitter
 mixed strategy given in \eqref{Eq:Trans_Optimal_Weak_Jammer}, if they have the same
 support and average jamming power.
\end{lem}

The outline of the proof for \eqref{Eq:J_AVE_M} will be given next.

\begin{proof}[(Outline)]
Assume $J_{Ave}$ is such that the transmitter is using $(m+1)$ of his pure strategies, i.e.,
\begin{equation*}
 \boldsymbol{x}^{*^T}_m = \begin{bmatrix} x_0 & x_1 & \dots & x_m & 0 & \dots & 0
 \end{bmatrix}_{1\times(N_T+1)}
\end{equation*}
where $\boldsymbol{x}^{*^T}_m$ is given in \eqref{Eq:Trans_Optimal_Weak_Jammer}. Using Lemma
\ref{Lem:N_T}, the jammer only needs to use the strategies  $J_j$ where $j = 0,\dots,$ $(m+1)$ and
the expected payoff of the game would be at least $R_{m+1}$ (otherwise the jammer had to use more
strategies). Lemma \ref{Lem:Semi_Uniform_Approx} suggests that the following semi-uniform
distribution which has the same support and average power as the jammer's optimal mixed strategy
could be used instead to compute the expected payoff of the game.
\begin{equation*}
 \boldsymbol{y}^T_{\text{SU}} = \begin{bmatrix} y_0 & y_1 & \dots & y_{m+1} & 0 & \dots & 0
 \end{bmatrix}_{1\times(N_T+1)}
\end{equation*}
\begin{equation}\label{Eq:Semi_Uniform_M}
 y_j =
  \begin{cases}
    1 - \frac{2N_T}{(m+2)}\cdot\frac{J_{Ave}}{J_{Max}} &\quad j = 0\\
    \frac{2N_T}{(m+1)(m+2)}\cdot\frac{J_{Ave}}{J_{Max}} &\quad j = 1, \dots,m + 1
  \end{cases}
\end{equation}
If we let the expected payoff of the transmitter be exactly $R_{m+1}$, then
\begin{equation}\label{Eq:R_M_Plus_One}
  \boldsymbol{x}^{*^T}_m C \boldsymbol{y}_{\text{SU}} = R_{m+1}
\end{equation}
Substituting \eqref{Eq:Trans_Optimal_Weak_Jammer} and \eqref{Eq:Semi_Uniform_M} in
\eqref{Eq:R_M_Plus_One} and solving for $J_{Ave}$ results in \eqref{Eq:J_AVE_M}.

Finally, letting $R_m = R_{N_T}$ or equivalently letting $m = (N_T-1)$ in \eqref{Eq:J_AVE_M} we
obtain the desired relation in \eqref{Eq:Power_Threshold}.
\end{proof}

For a weak jammer, the effective jamming power, $J_{\text{Eff}}$ is
\begin{equation}
 J_{\text{Eff}} = \left( \frac{m+1}{N_T} \right) J_{Max}
\end{equation}
If we define the effectiveness factor $E$ to be the ratio of the effective jamming power to the
actual average jamming power, we have
\begin{equation}
  \begin{aligned}
    E^{-1} & = \frac{J_{Ave}}{J_{\text{Eff}}}\\
	   & = \frac{\left( m + 1 - \alpha^{-1}_m R_{m+1} \right) \cdot \left( \frac{J_{Max}}{N_T}
\right)}{\left( \frac{m+1}{N_T} \right) J_{Max}}\\
	   & = 1 - \frac{1}{m+1} \alpha^{-1}_m R_{m+1} < 1\\
  \end{aligned}
\end{equation}
Similar to the case of the powerful jammer, the weak jammer can cause more damage to the
communication link than a Barrage noise jammer with an average power $J_{Ave}$.

\section{Continuous Case}\label{Sec:Continuous}

In this section we study the case where the jammer and the transmitter have infinite pure
strategies. In this case, instead of finite number of pure strategies, the transmitter and the
jammer have a continuum of pure strategies that could be represented as points in intervals $ R \in
\Big[ R(J_{Max}),  R(0) \Big]$ and $J \in \big[ 0, J_{Max} \big] $ respectively.

By letting $N_T \rightarrow \infty$ in \eqref{Eq:Power_Threshold}, we can find the jamming power
threshold for the continuous case to be
\begin{equation}
 \begin{aligned}
    J_{TH,\text{Lim}} = & \lim_{N_T\rightarrow \infty} J_{TH}\\
		      = &  J_{Max} - \frac{1}{2} \log \left( 1 + \frac{P_T}{N + J_{Max}} \right) \\
			& \times \int^{J_{Max}}_0 \left[ \frac{1}{2} \log \left( 1 + \frac{P_T}{N + J} \right) \right]^{-1} \cdot dJ\\
 \end{aligned}
\end{equation}

Similar to the discrete case, we can use a continuous semi-uniform distribution to approximate the
jammer's optimal mixed strategy and find an upper bound for $J_{TH,\text{Lim}}$.

\begin{equation}
 J_{TH,\text{Lim,UB}} = \frac{1}{2} \left[ 1 - \frac{R(J_{Max})}{R(0)} \right]J_{Max}
\end{equation}

%
%

\section{Conclusions}	\label{Sec:Conclusion}

We formulated the interaction between rate-adaptive communicating nodes and a smart power-limited jammer in a game-theoretic context. We show that packetization and adaptivity advantage the jammer. While, previous stationary information-theoretic capacity results correspond to a pure Nash-Equilibrium, packetized adaptive communication leads to lower game values. We show the existence of a mixed Nash Equilibrium and how to compute it. More importantly and surprisingly, we show the existence of a threshold on the average power of the jammer, above which the transmitter is forced to use a rate that corresponds to the maximum power of the jammer (and not the average power). We finally show how the optimal strategies can be computed and also derive a very simple (semi-uniform) jamming strategies that forces the transmitter to operate at the lowest rate (as if the jammer was continuously using its maximum power and not its average power).

%


%
%

\appendix

\section{Upper Bound for Jamming \\ Power Threshold}\label{App:I}

In Section \ref{SubSec:Power_Threshold} we showed that the 
\begin{equation}\label{Eq:Zi}
 J_{Ave} \geq Z_i = \frac{1}{2} J_{Max} \frac{N_T + 1}{N_T - i} \left(1 - \frac{R_{N_T}}{R_i} \right)
\end{equation}
where
\begin{equation}
 R_i = \frac{1}{2} \log \left( 1 + \frac{P_T}{N + \left( \frac{i}{N_T} \right) J_{Max}} \right)
\end{equation}
We stated that $Z_i$ given in \eqref{Eq:Zi} is strictly decreasing function of $i$ hence, $\max_i Z_I = Z_0 $ and an upper bound for the average power threshold is
\begin{equation}\label{Eq:JTHU}
 J_{TH,U} = \frac{1}{2} J_{Max} \frac{N_T + 1}{N_T} \left(1 - \frac{R_{N_T}}{R_0} \right)
\end{equation}
\begin{proof}
To prove \eqref{Eq:JTHU} first, we rewrite $Z_i$ as
\begin{equation}\label{Eq:Zi2}
  \begin{aligned}
    Z_i &=  \frac{1}{2} J_{Max} \frac{N_T + 1}{N_T} \frac{1}{1-\frac{i}{N_T}}  \left(1 - \frac{R_{N_T}}{R_i} \right)  			;\quad 0 \leq i < N_T\\
	&=  \left( \frac{1}{2} J_{Max} \frac{N_T + 1}{N_T} \right) \frac{J_{Max}}{J_{Max} - (\frac{i}{N_T} J_{Max})} \left(1 - \frac{R_{N_T}}{R_i} \right)\\
  \end{aligned}
\end{equation}
define $J$ and $R(J)$ as
\begin{equation}\label{Eq:J_and_RJ}
  \begin{aligned}
    & J = \left( \frac{i}{N_T} J_{Max} \right)  & 0 & \leq i < N_T  \\
    & R(J) = \frac{1}{2} \log \left( 1 + \frac{P_T}{N + J} \right) & 0 & \leq J < J_{Max} \\
  \end{aligned}
\end{equation}
substituting \eqref{Eq:J_and_RJ} in \eqref{Eq:Zi2} and we have
\begin{equation}\label{Eq:ZJ}
  \begin{aligned}
    Z(J) &= \left( \frac{1}{2} J_{Max} \frac{N_T + 1}{N_T} \right) \frac{J_{Max}}{J_{Max} - J} \left[1 - 			 	\frac{R(J_{Max})}{R(J)} \right]\\
	 &= a \times \frac{J_{Max}}{J_{Max} - J} \left[1 - \frac{R(J_{Max})}{R(J)} \right]\\
	 &= a\times F(J); \quad \text{where} \quad a > 0\ \text{and}\ 0\leq J < J_{Max}\\ 
  \end{aligned}
\end{equation}
If $F(J)$ in \eqref{Eq:ZJ} were a decreasing function of $J$ then, $Z_i$ and $Z(J)$ would also be decreasing functions of $i$ and $J$ respectively. Now let
\begin{equation}\label{Eq:FJ}
  F(J) = f(J)g(J) \quad \text{where} \quad 
  \begin{aligned}
    &f(J) = \frac{J_{Max}}{J_{Max} - J}\\
    &g(J) = 1 - \frac{R(J_{Max})}{R(J)}\\
  \end{aligned}
\end{equation}
For decreasing $F(J)$ we have
\begin{equation}\label{Eq:Decreasing_FJ}
 \begin{aligned}
    \frac{\partial}{\partial J}F & = g \frac{\partial}{\partial J}f + f \frac{\partial}{\partial J}g < 0\\
				 & f,g>0 \quad \text{for} \quad 0\leq J < J_{Max}  \Rightarrow \frac{\frac{\partial}{\partial J}f}{f} < - \frac{\frac{\partial}{\partial J}g}{g}\\
 \end{aligned}
\end{equation}
but from \eqref{Eq:FJ} we have 
\begin{equation}\label{Eq:dg_and_df}
  \begin{aligned}
 - \frac{\frac{\partial}{\partial J}g}{g} =& 
    \frac{1}{N+J} \times \left( \frac{x}{1 + x} \right) \left( \frac{\log (1 + x_m)}{\log (1 +x)} \right) \\
      & \quad \times \left( \frac{1}{\log (1 + x) - \log (1 + x_m) } \right) \\
  \frac{\frac{\partial}{\partial J}f}{f} =& \frac{P_T^{-1} x_m x}{x - x_m}\\
  \end{aligned}
\end{equation}
where
\begin{equation}\label{Eq:x}
 x = \frac{P_T}{N + J} \quad \text{and} \quad x_m = \frac{P_T}{N + J_{Max}} 
\end{equation}
obviously
\begin{equation*}
 0 < x_m < x
\end{equation*}
if we plug \eqref{Eq:dg_and_df} and \eqref{Eq:x} in \eqref{Eq:Decreasing_FJ} and simplify inequality we have
\begin{equation}
  \begin{aligned}\label{Eq:MathcalZ}
    \mathcal{Z} =& \frac{x_m^{-1}x}{1+x} \frac{x-x_m}{\log(1+x) - \log(1+x_m)} \frac{\log(1+x_m)}{\log(1+x)} \\
		 & > 1 \qquad \text{for} \quad 0 < x_m < x \\
  \end{aligned}
\end{equation}
We need to show that \eqref{Eq:MathcalZ} holds for all $ 0 < x_m < x $ but first, we notice that 
\begin{equation}\label{Eq:MathcalZ_Limit}
  \lim _{x \rightarrow x_m^{+}} \mathcal{Z} \sim \frac{x_m^{-1} \log(1+x_m)}{x^{-1} \log(1+x)} \rightarrow 1^+ \quad \forall \ 0 < x_m < x
\end{equation}
since we have
\begin{equation}
  \frac{\partial}{\partial z} z^{-1} \log(1+z) < 0 \quad \forall \ 0 < z
\end{equation}
where we used the following natural logarithm property
\begin{equation}\label{Eq:Log_Inequality}
  \frac{z}{1+z} < \log(1+z) \leq z \quad \text{for all} \quad z>0 
\end{equation}
For simplicity we rewrite inequality in \eqref{Eq:MathcalZ} as
\begin{equation}\label{Eq:MathcalZ2}
  \begin{aligned}
 \mathcal{Z}_2 =& \left[ x(x-x_m) \log(1+x_m) \right] \\ 
	       -& \left[ x_m (1+x) \log(1+x) \left( \log(1+x) - \log(1+x_m) \right) \right] \\
		& > 0
  \end{aligned}
\end{equation}
As a result of \eqref{Eq:MathcalZ_Limit} we have $ \lim_{x \rightarrow x_m^{+}} \mathcal{Z}_2 \rightarrow 0^+ $ for all $0<x_m<x$. Since \eqref{Eq:MathcalZ2} holds for $x\rightarrow x_m^+ $, if $\mathcal{Z}_2$ was strictly increasing function of $x$ for all $x>x_m$, \eqref{Eq:MathcalZ2} and \eqref{Eq:MathcalZ} would also hold as a corollary.

To show that $\mathcal{Z}_2$ is strictly increasing, we first verify that 
\begin{equation}\label{Eq:MathcalZ2_Derivative}
 \frac{\partial \mathcal{Z}_2}{\partial x} \bigg( x = x_m \bigg) = 0
\end{equation}
given that \eqref{Eq:MathcalZ2_Derivative} is true, an alternative way to proceed is to show that $\frac{\partial \mathcal{Z}_2}{\partial x}$ is itself strictly increasing function of $x$ (to show that the second partial derivate is strictly positive). Define $\mathcal{Z}_3$
\begin{equation}\label{Eq:MathcalZ3}
  \begin{aligned}
      \mathcal{Z}_3 =& \frac{\partial^2 \mathcal{Z}_2}{\partial x^2} \times (1+x)\\
		    =& 2\log(1+x_m) - 2x_m + 2x \log(1+x_m)\\
		     & - 2x_m \log(1+x) + x_m \log(1+x_m)\\
  \end{aligned}
\end{equation}
It can be verified that for all $x>x_m$ and $x_m > 0 $ we have  $ \lim_{x \rightarrow x_m^{+}} \mathcal{Z}_3 > 0 $. Taking the partial derivate of $\mathcal{Z}_3$ with respect to $x$ we have
\begin{equation}
      \frac{\partial \mathcal{Z}_3}{\partial x} = 2 \left[ \log(1+x_m) - \frac{x_m}{1+x} \right]
\end{equation}
but from \eqref{Eq:Log_Inequality} we have
\begin{equation}
  \begin{aligned}
      &\log(1+x_m) > \frac{x_m}{1+x_m}>\frac{x_m}{1+x} \quad \text{for all} \quad x>x_m\\
      & \Rightarrow \quad 2 \left[ \log(1+x_m) - \frac{x_m}{1+x} \right] >0 \quad \forall x>x_m>0\\
  \end{aligned}
\end{equation}
and hence we conclude that
\begin{equation}
 \frac{\partial \mathcal{Z}_3}{\partial x} > 0 \qquad \text{for all} \quad x>x_m>0
\end{equation}
and $\mathcal{Z}_2$ is indeed an increasing function of $x$ for all $0<x_m<x$. Taking the reverse steps that resulted in \eqref{Eq:MathcalZ2} and \eqref{Eq:MathcalZ} we conclude that $Z_i$ in \eqref{Eq:Zi} is indeed a decreasing function and hence $J_{TH,U}$ given in \eqref{Eq:JTHU} is an upper bound for $J_{TH}$ in \eqref{Eq:Power_Threshold}.
\end{proof}

\section{Optimal Mixed-Strategies}\label{App:II}

Assume $0 \geq J_{Ave} < J_{TH} $ is such that jammer's optimal mixed strategy is to use $(m+1)$ of his pure strategies. It is easy to show that in such a case, the transmitter's optimal mixed strategy includes, at most, $(m+1)$ of his pure strategies. For now, we assume the transmitter is using $(m)$ of his pure strategies, i.e.;
\begin{equation}\label{Eq:Mixed_Strategy_Vectors}
 \begin{aligned}
    \boldsymbol{x}^T &= \begin{bmatrix}
                       x_0 & x_1 & \dots & x_m & 0 & \dots & 0
                      \end{bmatrix}_{1\times (1+N_T)}\\
    \boldsymbol{y}^T &= \begin{bmatrix}
                       y_0 & y_1 & \dots & y_m & y_{m+1} & 0 & \dots & 0
                      \end{bmatrix}_{1\times (1+N_T)}\\
		     & \qquad  \text{where} \quad 0 \leq m < N_T\\
 \end{aligned}
\end{equation}
The expected payoff of the game for the mixed-strategy pair $(\boldsymbol{x},\boldsymbol{y}) $ given in \eqref{Eq:Mixed_Strategy_Vectors} is
\begin{equation*}
  \begin{aligned}
  C(\boldsymbol{x},\boldsymbol{y}) &= \boldsymbol{x}^T C \boldsymbol{y}\\
				  &= \begin{bmatrix}
				      x_0 & x_1 & \dots & x_m & 0 & \dots & 0
				     \end{bmatrix}\\
				  & \times  \begin{bmatrix}
					      R_0     & 0       & 0     & \dots   & 0      \\
					      \vdots  & \ddots  & 0     & 0       & \vdots \\
					      R_i     & \dots   & R_i   & 0       & \vdots \\
					      \vdots  &         &       & \ddots  & \vdots \\
					      R_{N_T} & R_{N_T} & \dots & R_{N_T} & R_{N_T}\\
					    \end{bmatrix}
				      \begin{bmatrix}
					y_0 \\ y_1 \\ \dots \\ y_m \\ y_{m+1} \\ 0 \\ \dots \\ 0
				      \end{bmatrix}\\
  \end{aligned}
\end{equation*}
\begin{equation*}
 \begin{aligned}
		& = \Bigg[ \sum^{m}_{i=0} x_iR_i \quad \sum^{m}_{i=1} x_iR_i \quad \dots \\
		& \quad \sum^{m}_{i=j} x_iR_i \quad \dots \quad \sum^{m}_{i=m} x_iR_i \quad 0 \quad \dots \Bigg]
				      \begin{bmatrix}
					y_0 \\ y_1 \\ \dots \\ y_m \\ y_{m+1} \\ 0 \\ \dots \\ 0
				      \end{bmatrix}\\  
 \end{aligned}
\end{equation*}
\begin{equation}\label{Eq:General_Expected_Payoff}
 \begin{aligned}
	=& y_0 \sum^{m}_{i=0} x_iR_i + y_1 \sum^{m}_{i=1} x_iR_i + \dots \\ 
	 & \qquad + y_j \sum^{m}_{i=j} x_iR_i + \dots + y_mx_mR_m\\  
 \end{aligned}
\end{equation}
We can rewrite \eqref{Eq:General_Expected_Payoff} in terms of $R_i$'s;
\begin{equation}
  \begin{aligned}
  C(\boldsymbol{x},\boldsymbol{y}) = & x_0 R_0 (y_0) + x_1 R_1 (y_0 + y_1) + \dots \\
				    &  \quad + x_i R_i \sum^{i}_{j=0} y_j + \dots x_m R_m \sum^{m}_{j=0} y_j
  \end{aligned}
\end{equation}
Assume the transmitter is using the mixed-strategy $\widehat{\boldsymbol{x}}$, defined below, which not necessary an optimal mixed-strategy.
\begin{equation}\label{Eq:x_Hat_Vector}
 \widehat{\boldsymbol{x}} = \alpha^{-1}_m \begin{bmatrix} R_0^{-1}&R_1^{-1}&\dots&R_m^{-1}&0&\dots&0 \end{bmatrix} _{1\times (N_T + 1)}
\end{equation}
where $\alpha_m$ is defined as
\begin{equation}\label{Eq:Alpha_m}
 \alpha_m^{-1} = \sum^{m}_{i=0} R_i^{-1} \qquad 0<m<N_T
\end{equation}
The expected payoff of the game for the mixed-strategy pair $(\widehat{\boldsymbol{x}},\boldsymbol{y})$ is
\begin{equation}\label{Eq:Expected_Payoff_for_x_Hat}
 \begin{aligned}
    C \big(\widehat{\boldsymbol{x}},\boldsymbol{y} \big) =& \widehat{\boldsymbol{x}}^T C \boldsymbol{y}\\
					       =& \Bigg[ (y_0) + (y_0 + y_1) + \dots\\
						& \quad + \left( \sum^{i}_{j=0} y_j \right) + \dots + \left( \sum^{i}_{j=m} y_j \right) \Bigg] \alpha_m\\
					       =& \Bigg[ \left( 1 - \sum^{m+1}_{j=1} y_j \right) + \left( \sum^{m+1}_{j=2} y_j \right) + \dots \\
						& \quad + \left( 1- \sum^{m+1}_{j=i+1} y_j \right) + \dots + \left( 1 - y_m \right) \Bigg] \alpha_m\\
 \end{aligned}
\end{equation}
By expanding the sums in \eqref{Eq:Expected_Payoff_for_x_Hat} we can rewrite the expected payoff of the game for $(\widehat{\boldsymbol{x}},\boldsymbol{y})$ in a more compact form
\begin{equation}\label{Eq:Expected_Payoff_for_x_Hat_Compact}
  \begin{aligned}
    &C \big(\widehat{\boldsymbol{x}},\boldsymbol{y} \big) = \alpha_m \times\\
    &
  \begin{pmatrix}
    1 -    & (      & y_1 &  + y_2 & + & \dots  & + y_{i+1} & \dots & + y_{m+1} )+\\
    1 -    & (      &     &\ \ y_2 & + & \dots  & + y_{i+1} & \dots & + y_{m+1} )+\\
    \vdots & \vdots &     &        &   & \ddots &   \vdots  &       &     \vdots  \\
    1 -    & (      &     &        &   &        & + y_{i+1} & \dots & + y_{m+1} )+\\
    \vdots & \vdots &     &        &   &        &           & \ddots&     \vdots  \\
    1 -    & (      &     &        &   &        &           & \dots & + y_{m+1} )\ \ \\
  \end{pmatrix} \\
    &= \left[ \big( m+1 \big) - \left( \sum^{m+1}_{j=0} jy_j \right)  \right] \alpha_m
  \end{aligned}
\end{equation}
but from \eqref{Eq:Jammer_Mixed_Set_1} and for all $\boldsymbol{y} \in  \mathbb{Y}_{\text{E}} $ we have
\begin{equation}\label{Eq:Average_Power_Constraint}
 \begin{aligned}
    \boldsymbol{J}^T \boldsymbol{y} = J_{Ave} &= \sum^{m+1}_{j=0} \left( \frac{j}{N_T} J_{Max} \right) y_j \\
					      &= \left( \sum^{m+1}_{j=0} j y_j \right) \frac{J_{Max}}{N_T}\\
			\Rightarrow \qquad    & \sum^{m+1}_{j=0} j y_j = N_T \frac{J_{Ave}}{J_{Max}}
 \end{aligned}
\end{equation}
substituting \eqref{Eq:Average_Power_Constraint} in \eqref{Eq:Expected_Payoff_for_x_Hat_Compact} and the expected payoff of the game for the the mixed-strategy pair $(\widehat{\boldsymbol{x}}, \boldsymbol{y})$ becomes
\begin{equation}\label{Eq:Expected_Payoff_for_x_Hat_Final}
 C \big(\widehat{\boldsymbol{x}}, \boldsymbol{y}\big) = \left( m+1 - N_T \frac{J_{Ave}}{J_{Max}} \right) \alpha_m
\end{equation}
Hence, by using $\boldsymbol{\widehat{x}}$ (which is not necessary an optimal mixed strategy) against jammer's arbitrary mixed-strategy with average power $J_{Ave}$, the transmitter can achieve the expected payoff given in \eqref{Eq:Expected_Payoff_for_x_Hat_Final}. Therefore, the expected payoff of the game at equilibrium must at least be equal to \eqref{Eq:Expected_Payoff_for_x_Hat_Compact}, i.e.,
\begin{equation}\label{Eq:Payoff_for_Trans_Optimal}
  C \big(\boldsymbol{x}^*, \boldsymbol{y}^*  \big) \geq C \big(\widehat{\boldsymbol{x}}, \boldsymbol{y}^* \big) = \left( m+1 - N_T \frac{J_{Ave}}{J_{Max}} \right) \alpha_m
\end{equation}
In the same way, it can be shown that if the transmitter and the jammer were using the same number of pure strategies, $(m+1)$, the mixed-strategy given in \eqref{Eq:x_Hat_Vector} results in the same expected payoff give that $m$ is replaced by $(m+1)$.  


Now, assume $J_{Ave}$ is such that the jammer is using $ (m+1)$ of his pure strategies. Define the the following mixed-strategy for the jammer which is not necessary an optimal mixed-strategy.
\begin{multline}\label{Eq:y_Hat}
 \widehat{\boldsymbol{y}}^T = \begin{bmatrix}
                       y_0 & y_1 & \dots & y_m & y_{m+1} & 0 & \dots & 0
                      \end{bmatrix}_{1\times (1+N_T)}\\
  y_j =  \begin{cases}
                 R_0^{-1}R_{m+1}&  j=0\\
		 \Big( R_j^{-1} - R_{j-1}^{-1} \Big)R_{m+1} &  j=1,\dots, m+1\\
		 0 &  m+1<j\leq N_T\\
                \end{cases} \\
\end{multline}
it can be verified that $\widehat{\boldsymbol{y}}$ is indeed a mixed-strategy vector;
\begin{equation}
 \begin{aligned}
    \sum^{m+1}_{j=0}y_j & =  R_{m+1}R_0^{-1} + R_{m+1} \sum^{m+1}_{j=1} \Big( R_j^{-1} - R_{j-1}^{-1} \Big)\\
			& =  R_{m+1} \left( R_0^{-1} + \sum^{m+1}_{j=1} R_j^{-1} - \sum^{m+1}_{j=1} R_{j-1}^{-1} \right) \\ 
			& =  R_{m+1} \left( R_{0}^{-1} + R_{m+1}^{-1} - R_{0}^{-1} \right) = 1\\
 \end{aligned}
\end{equation}
Furthermore, we have
\begin{equation}
 \begin{aligned}
    \sum^{m+1}_{j=0} jy_j &= \sum^{m+1}_{j=1} jy_j = R_{m+1} \left( \sum^{m+1}_{j=1} j \Big( R_j^{-1} - R_{j-1}^{-1} \Big) \right)\\
			  & = R_{m+1} \left( \sum^{m+1}_{j=1} jR_j^{-1} - \sum^{m+1}_{j=1} \big(j-1+1\big)R_{j-1}^{-1} \right) \\
			  &= R_{m+1} \Bigg( \sum^{m+1}_{j=1}jR_j^{-1} - \sum^{m+1}_{j=1} (j-1)R_{j-1}^{-1}  \\
			  & \qquad \qquad - \sum^{m}_{j=0} R_{j}^{-1} \Bigg)\\
			  &= R_{m+1} \left( \big(m+1\big) R_{M=1}^{-1} - \sum^{m}_{j=0} R_j^{-1} \right)\\
			  &= \big(m+1\big) - \alpha_m^{-1} R_{m+1}\\
 \end{aligned}
\end{equation}
hence the average power of $\widehat{\boldsymbol{y}}$ becomes;
\begin{equation}\label{Eq:y_Hat_Average_Power}
 \boldsymbol{J}^T \widehat{\boldsymbol{y}} = J_{Ave,m} = \Big( m+1 - \alpha_m^{-1} R_{m+1} \Big) \frac{J_{Max}}{N_T}
\end{equation}
Assuming the jammer is using $\widehat{\boldsymbol{y}}$ against transmitter's arbitrary mixed-strategy, the expected payoff of the game is
\begin{equation}
  \begin{aligned}
 & C \big( \boldsymbol{x}, \widehat{\boldsymbol{y}} \big) =  \boldsymbol{x}^T C \widehat{\boldsymbol{y}}
							= R_{m+1} \boldsymbol{x}^T \times \\
					    &\begin{bmatrix}
					      R_0     & 0       & 0     & \dots   & 0      \\
					      \vdots  & \ddots  & 0     & 0       & \vdots \\
					      R_i     & \dots   & R_i   & 0       & \vdots \\
					      \vdots  &         &       & \ddots  & \vdots \\
					      R_{N_T} & R_{N_T} & \dots & R_{N_T} & R_{N_T}\\
					    \end{bmatrix}
					    \begin{bmatrix}
					      R_0^{-1}\\
					      R_1^{-1} - R_2^{-1}\\
					      \vdots\\
					      R_{j}^{-1} - R_{j-1}^{-1}\\
					      \vdots\\
					      R_{m+1}^{-1} - R_m^{-1}\\
					      0\\
					      \vdots\\
					      0
					    \end{bmatrix}\\
  &\qquad = R_{m+1} \boldsymbol{x}^T \begin{bmatrix}
                                      \boldsymbol{1}_{(m+2)\times 1}\\
				      0\\
				      \vdots\\
				      0
                                     \end{bmatrix} = R_{m+1}
  \end{aligned}
\end{equation}
since $ \boldsymbol{x}$ is a mixed-strategy and has at most $(m+1)$ non-zero elements (see \eqref{Eq:Mixed_Strategy_Vectors}).

Therefore, by using $\widehat{\boldsymbol{y}}$ given in \eqref{Eq:y_Hat} against transmitter's arbitrary mixed-strategy, the jammer guarantees not to lose more than  $R_{m+1}$ given that his average jamming power is $J_{Ave,m}$ given in \eqref{Eq:y_Hat_Average_Power}. Since $\widehat{\boldsymbol{y}}$ is not necessary an optimal mixed-strategy for the jammer, the optimal mixed-strategy would at most be less than $R_{m=1}$, i.e.,
\begin{equation}\label{Eq:Game_Upper_Bound}
 C \big( \boldsymbol{x}^*, \boldsymbol{y}^* \big) \leq C \big( \boldsymbol{x}^*, \widehat{\boldsymbol{y}} \big) = R_{m+1} 
\end{equation}
It can be shown (by induction) that for specific values of average jamming power given by \eqref{Eq:y_Hat_Average_Power} and for $ m = 0,\dots,N_T-1 $, optimal mixed-strategy for the transmitter is to use $(m+1)$ of his pure strategies. From \eqref{Eq:Payoff_for_Trans_Optimal} and by letting $m \rightarrow (m+1) $ we have 
\begin{equation}\label{Eq:Game_Lower_Bound}
 \begin{aligned}
    C\big(\boldsymbol{x}^*,\boldsymbol{y}^*\big) \geq & \left(m+2 - N_T \frac{J_{Ave,m}}{J_{Max}} \right) \alpha_{m+1}\\
						      & \quad = \Big(1 + \alpha_m^{-1} R_{m+1} \Big) \alpha_{m+1}\\
    \alpha_{m+1}^{-1} = \alpha_m^{-1} + R_{m+1}^{-1} \Rightarrow & \quad = \Big( 1 + \alpha_{m+1}^{-1} R_{m+1} -1 \Big) \alpha_{m+1}\\
						      & \quad = R_{m+1}\\
 \end{aligned}
\end{equation}
Therefore from \eqref{Eq:Game_Upper_Bound} and \eqref{Eq:Game_Lower_Bound} we have
\begin{multline}
 R_{m+1} = C\big(\widehat{\boldsymbol{x}},\boldsymbol{y}^*\big) \leq C\big(\boldsymbol{x}^*,\boldsymbol{y}^* \big) \leq C\big(\boldsymbol{x}^*,\widehat{\boldsymbol{y}}\big) = R_{m+1}\\
  \Rightarrow C\big(\boldsymbol{x}^*,\boldsymbol{y}^* \big) = R_{m+1} \quad \text{for} \quad J_{Ave} = J_{Ave,m}
\end{multline}
and hence $\widehat{\boldsymbol{x}}$ and $\widehat{\boldsymbol{y}}$ defined in \eqref{Eq:x_Hat_Vector} and \eqref{Eq:y_Hat} are indeed optimal mixed-strategies for the transmitter and the jammer respectively\footnote{Not necessary unique though.}.


\end{document}